\documentclass[twoside,11pt, a4paper]{article}

\usepackage[utf8]{inputenc}
\usepackage[T1]{fontenc}
\usepackage[english]{babel}

\usepackage[usenames,dvipsnames]{xcolor}
\usepackage{pgf,tikz}
\usetikzlibrary{arrows}
\usetikzlibrary{decorations.markings,decorations.pathmorphing,patterns}
\usepackage{pgfplots}

\usepackage{booktabs,bm,multirow}%
\usepackage{enumitem}

\usepackage{geometry}
\usepackage{amsmath,amsfonts,amssymb,amsthm}%
\usepackage{dsfont,multicol,marvosym}

\usepackage{subcaption}

\renewcommand{\phi}{\varphi}
\renewcommand{\theta}{\vartheta}
\renewcommand{\epsilon}{\varepsilon}

\newcommand{\R}{\mathbb{R}}
\newcommand{\Z}{\mathbb{Z}}
\newcommand{\N}{\mathbb{N}}
\newcommand{\CC}{\mathcal{C}}
\newcommand{\PP}{\mathcal{P}}
\DeclareMathOperator{\dist}{dist}
\DeclareMathOperator{\graph}{graph}

\newcommand{\abs}[1]{\left\lvert#1\right\rvert}                  
\newcommand{\norm}[1]{\left\lVert#1\right\rVert}    
\newcommand*{\dd}{\mathop{}\!\mathrm{d}}

\newtheorem{theorem}{Theorem}
\newtheorem{lemma}[theorem]{Lemma}
\newtheorem{prop}[theorem]{Proposition}
\newtheorem{corol}[theorem]{Corollary}
\theoremstyle{definition}
\newtheorem{defin}[theorem]{Definition}
\newtheorem{example}[]{Example}

\numberwithin{equation}{section}
\numberwithin{theorem}{section}
\numberwithin{example}{section}

\title{Limit cycles for dynamic crawling locomotors with periodic prescribed shape}

\date{\emph{Preprint version: 11 October 2022}}
\author{Paolo Gidoni\footnote{Institute of Information Theory and Automation of the Czech Academy of Sciences, Pod vodárenskou věží 4, CZ-182 00 Prague 8, Czechia; \emph{email:} \texttt{gidoni@utia.cas.cz}}, \
Alessandro Margheri\footnote{Centro de Matemática, Aplicações Fundamentais e Investigação Operacional, Departamento de Matemática, Faculdade de Ciências, Universidade de Lisboa, Campo Grande, Edifício C6, piso 2,  1749-016  Lisboa, Portugal; \emph{email:} \texttt{amargheri@fc.ul.pt}} \
and Carlota Rebelo\footnote{Centro de Matemática Computacional e Estocástica, Departamento de Matemática, Faculdade de Ciências, Universidade de Lisboa,  Campo Grande, Edificio C6, piso 2, 1749-016 Lisboa,  Portugal; \emph{email:} \texttt{mcgoncalves@fc.ul.pt}}}

\usepackage{fancyhdr}
\begin{document}
\maketitle
\begin{abstract}
We study the asymptotic behaviour of a family of dynamic models of crawling locomotion, with the aim of characterizing a gait as a limit property.
The locomotors, which might have a discrete or continuous body, move on a line with a periodic prescribed shape change, and might possibly be subject to external forcing (e.g., crawling on a slope). We discuss how their behaviour is affected by different types of friction forces, including also set-valued ones such as dry friction.
We show that, under mild natural assumptions, the dynamics always converge to a relative periodic solution.  The asymptotic average velocity of the crawler yet might still depend on its initial state, so we provide additional assumption for its uniqueness. In particular, we show that the asymptotic average velocity is unique both for strictly monotone friction forces, and also for dry friction, provided in the latter case that the actuation is sufficiently smooth (for discrete models) or that the friction coefficients are always nonzero (for continuous models). We present several examples and counterexamples illustrating the necessity of our assumptions.

\end{abstract}

\noindent {\small
 }

\section{Introduction}
A periodic pattern of shape changes is the keystone in the description of most biological and robotic locomotion  strategies.
Common examples are the flapping of fins in fishes or wings in birds, a peristaltic wave in a crawler, as well as the rotations of a ship's propeller. Periodicity brings a clear design advantage, since a long-term and complex task, such as going from A to B, can be decomposed as the iteration of a much simpler and brief input.
Indeed, not only a periodic input can be easily implemented in a robotic device, but also several biological organisms are known to employ basic mechanisms called \emph{central pattern generators} to produce cyclic shape-change patterns, without recurring to any movement-related sensory feedback \cite{Ijs08}.

To fully describe a locomotion strategy, identifying what is usually called a \emph{gait}, a specific periodic shape-changing pattern must be associated with a corresponding movement of the locomotor in a given environment.
More abstractly, a gait can be identified as a \emph{relative-periodic} evolution of the system \cite{KeMu,FaPaZo}. 

Relative-periodicity is based on the decomposition of the configuration space of the locomotor into the product of a \emph{shape space} and a \emph{position space}. 
The shape space describes, as predictable, the shape of the body of the locomotor and is where a periodic evolution is expected.
In the models discussed in our paper, the shape space will be an Euclidean space in discrete models and a Sobolev  space in continuous ones; but, for instance, in the presence of a rotating body part a manifold, such as $\R^n\times \mathcal S^1$, would a  suitable choice. 

The position space describes the location and orientation of the locomotor and, usually, is a Lie group. To each gait, we would like to  associate an element $\gamma$ of the position space, called \emph{geometric phase} whose action on the group describes the movement of the locomotor. Given an initial position $y_0$, a single iteration of the gait will displace the locomotor to $\gamma y_0$.

Unfortunately, while relative-periodicity is an extremely useful structure to study locomotion,  it is also a sort of ideal behaviour that can be consistently observed only in a limited set of models.
Examples where periodic inputs always produce relative-periodic evolutions are swimming at low-Reynolds number \cite{LaugaBook} and some special models of wheeled locomotion \cite{KeMu} and of crawling \cite{Rehor,Ago, DeSTat12}.

Often, instead, a relative-periodic behaviour might be expected to emerge as an asymptotic behaviour of the system. The Reader might get an intuitive description of this phenomenon by considering a degenerate example of crawling: a passive object lying on a surface, subject to friction but without any actuation. 
It is easy to identify its associated relative-periodic behaviour: a constant shape and a stationary position, so that $\gamma$ is the identity of the group. However, this behaviour is instantly reached only if the initial condition is also stationary.
An analogous phenomenon is produced by an elastic body with an initial deformation different from the rest configuration.
In such cases, the stationary asymptotic behaviour can be showed by noticing that the mechanical energy of the system  is a  Lyapunov function, decreasing in time. 
The situation, however, becomes much  more challenging if we consider a true locomotion model, since, due to the work produced by the actuation, the energy of the system is no longer decreasing. 

The aim of this paper is to rigorously investigate such an asymptotic behaviour for some general families of models of dynamic crawling locomotion with prescribed shape. We emphasize that our work is not limited to proving the convergence of each solution of the system to a relative periodic behaviour. For practical applications, it is also necessary to show that the (asymptotic) shape change and geometric phase do not depend on the initial conditions. As we will show, counterexamples are possible and the uniqueness of the limit behaviour requires stronger assumptions than convergence alone.

This kind of investigation is pivotal in the design, operation and optimization of robotic devices. As we mentioned, the possibility to rely  just on a small toolbox of periodic patterns, without any dedicated sensor or feedback mechanism, reduces the complexity of the device, with advantages for manufacturing, cost and miniaturization. This is however possible only if the behaviour of the locomotor is not affected in a relevant way by uncontrolled factors, such as the state of the locomotor when a new gait is applied, or a temporary external perturbation. The fact that a locomotion strategy do not involve a full control of the state of the locomotor can be interpreted as a basic form of \emph{morphological computation} \cite{MuHo}, meaning that the ability to adapt to the external conditions is partially delegated to structural properties of the robot, reducing the complexity of the actuation.

 A rigorous proof of the well-posedness and uniqueness of the asymptotic behaviour of the system is not only relevant \emph{per se}, but can  also be seen as a preliminary step to gait optimization.
Indeed, since gaits can be properly defined only as an asymptotic property of a periodic input, gait optimization must also be evaluated in the long-time limit \cite{GirJean}, so that optimality is considered only among limit-cycles. Example of this limit-cycle optimization have been discussed, for instance, for the Chaplygin sleigh \cite{FedTal20,PoFeTa} and for a model of quasistatic crawling \cite{ColGid}.

Whereas asymptotic stabilization is usually observed numerically or experimentally for a specific choice of the actuation (e.g. \cite{WagLau}), in this paper we undertake a more theoretical approach, providing rigorous results on the existence and structure of a global attractor for the dynamics.
Such an analytical approach allows in general to explore the effects (or lack thereof) of the different elements in a locomotion model: the actuation pattern, the geometry of the locomotor, the rheology of the interaction with the substrate, the relevance of inertial effects and the possible elasticity of the body.
In this paper our focus will be on the effect of the rheology in a dynamic framework for general actuation strategies.
Thus, we will restrict ourselves to rectilinear models of crawlers -- although highlighting the qualitative differences between discrete and continuous models -- and neglect elastic deformations, so that the shape of the crawler is directly controlled by the actuation. 

A prescribed shape allows to  reduce the dynamics to the position space $\R$. In particular, an asymptotically relative-periodic behaviour corresponds to a limit cycle for the velocity $v=\dot{\bar x}$ of the position of the locomotor. Accordingly, the geometric phase $\gamma$ identifies the asymptotic limit of the Poincaré time-map, with its value describing the \emph{asymptotic average velocity} $\gamma/T$ of the locomotor.

For this reason, we start in Section \ref{sec:abstr} by discussing some general results for a special class of scalar time-periodic differential inclusions $\dot v\in G(t,v)$. This framework allows to deal also with set-valued friction forces such as dry friction. For  the class considered we  prove that the dynamics is asymptotically periodic and that the periodic limit of a solution lies on a global attractor. Strengthening in two alternative ways the   assumptions on $G$,  we prove that such attractor  is a unique  limit cycle.

In Section \ref{sec:discr}  we apply these general  results to  discrete models of crawling locomotion, considering various friction laws. We show that, under mild dissipativity assumptions, the system will always converge to a relative-periodic behavior, which however might depend on the initial state. The uniqueness of the asymptotic velocity is obtained either for strictly monotone friction forces, extending the results in \cite{FigKny}, or for dry friction if the actuation is sufficiently smooth in time (continuous friction coefficients and a $\CC^1 $ shape change). Several examples and counterexamples are included, illustrating the sharpness of our assumptions.

In Section \ref{sec:cont} we repeat the same analysis for continuous models of crawlers, obtaining analogous results. The  only difference is in the case of dry friction, for which uniqueness of the limit cycle does not require any additional time-regularity. We remark that, both in the discrete and in the continuous case, dry friction provides only the (weak) monotonicity of $G$ in $v$, that is not by itself sufficient for uniqueness: its proof relies specifically on the intrinsic structure of our locomotion models.

Finally, in Section  \ref{sec:disc}  we discuss the results obtained in the paper in the context of the existing literature, together with possible future developments.

\section{Theoretical results for first order differential equations and inclusions} \label{sec:abstr}

\subsection{Structural assumptions}
Since we plan to describe the locomotion of crawlers under  general friction forces, including also dry friction, our framework should accomodate also  set-valued maps. Hence, we denote by $\PP(E)$ the power set of a set $E$. Moreover, with a slight abuse of notation, given $B\in \PP(\R)$ and $a\in \R$, we identify $B+a\in\PP(\R)$ as the set of the elements $b+a$ with $b\in B$.
Let us also recall some monotonicity properties for set-valued maps $F\colon \R \to \PP(\R)$.

\begin{defin}
We say that the set-valued map $F\colon \R \to \PP(\R)$ is \emph{monotone increasing} (resp. \emph{decreasing})  if 
\begin{equation*}
(y_2 - y_1)(u_2-u_1)\geq 0 \qquad \text{for every $(u_1,u_2)\in\R^2, u_1\neq u_2, y_1\in F(u_1), y_2\in F(u_2)$}
\end{equation*}
(resp. if $(y_2 - y_1)(u_2-u_1)\leq 0 $ on the same domain).

We say that a monotone  map is \emph{strictly monotone}   if the corresponding inequality is always strict. \end{defin}

\begin{defin}
We say that a set-valued monotone map $F\colon \R \to \PP(\R)$   is \emph{maximal monotone} \cite{Brezis}  if it is maximal among the set of monotone set-valued maps, with the respect to the relation of graph inclusion, i.e. if there is no monotone map $\widetilde F\colon \R \to \PP(\R)$ such that $\graph F \subsetneqq \graph \widetilde F$.
\end{defin}
Notice that, given a convex function $V\colon \R\to \R$ then its subdifferential  $\partial V$   is a maximal monotone decreasing map, with  strict monotonicity corresponding to a strictly convex $V$.  

Set-valued friction forces, however, also bring a favourable structure to the problem, providing, for instance, existence and forward-in-time uniqueness of solution.
For this reason, we will focus on a special class of differential inclusions $\dot u\in G(t,u)$ satisfying the following assumptions on $G$, which  will be referred  to in the rest of the paper as \textit{structural assumptions}: 

\begin{enumerate}[label=\textup{(S\arabic*)} ]  
\item  \label{cond:S1}\textit{The set-valued map $G\colon\R\times \R\to \PP(\R)$ is of the following form}
\begin{equation}\label{eq:Gform}
G(t,u)=-\mathcal{A}(t,u)+p(t,u)
\end{equation}
where
\begin{itemize}
\item $\mathcal{A}\colon\R\times \R\to \PP(\R)\setminus\emptyset$ is a set-valued map, locally bounded by a measurable function, $T$-periodic in $t$, maximal monotone increasing  in $u$ for every $t$ and measurable in $t$ for every $u$. 
	\item \textit{ $p\colon \R\times\R\to \R$ is a single-valued Carathéodory function, $T$-periodic in $t$, locally Lipschitz continuous in the variable $u$ uniformly in $t$.}
\end{itemize}
	\item \label{cond:S2}  \textit{There exist  a positive constant $R>0$ and two  $T$-periodic measurable functions $\ell_d^\pm(t)\colon\R\to\R$ such that $\int_0^T\ell_d^\pm(s)\dd s<0$ and }
	\begin{align*}
		y&\geq-\ell_d^-(t) &&\textit{for almost every $t\in \R, u\in(-\infty, -R],\,y\in G(t,u)$,}\\[2mm]
				y&\leq\ell_d^+(t) &&\textit{for almost every $t\in \R, u\in[R,+\infty),\,y\in G(t,u)$.}
	\end{align*}
\end{enumerate}

Notice that  in \ref{cond:S1} we require $G$ to be nowhere empty-valued. Moreover, for every $t$,  $G(t,\cdot)$   is convex- and  compact-valued, upper semi-continuous, single-valued  outside  a set of null measure and locally Lipschitz continuous in $u$.  

\smallskip

We plan to study the differential problem
\begin{equation}\label{eq:general}
	\dot u\in G(t,u)\,.
\end{equation}
We will call \emph{solution} of \eqref{eq:general} an absolutely continuous function such that \eqref{eq:general} is satisfied at almost every $t$ on the domain of the solution.
In particular, we observe that we have local existence \cite[Corollary~5.2, pag.~59]{Deim} and right uniqueness (using the same argument as \cite[Theorem~1, pag.~106]{Filip}) of solution for the Cauchy problems associated with \eqref{eq:general}; see also \cite{VilNgu}.


\subsection{Existence of periodic solutions and attractors}

We now investigate the existence and qualitative properties of periodic solutions and global attractors for the dynamics of  \eqref{eq:general}, assuming that $G$ satisfies the structural assumptions  \ref{cond:S1} and \ref{cond:S2}. 

The existence of periodic solution and attractors has been studied by several authors in the generalized framework of  Hilbert spaces, both in the cases when the maximal monotone term is autonomous \cite{Hir,Fri, AkaSte} and time-periodic \cite{Ken,Ota,PapRad}.
For our purposes, we restrict ourselves to the scalar case. If, on one hand, this simplifies the problem, on the other hand it allows us to relax the assumptions on the system and to obtain additional qualitative properties of the set of periodic solutions. For instance, the scalar setting together with the structure of the crawling problem will allow us to obtain uniqueness of the periodic solution also for a monotone $G$ (Theorem~\ref{th:Cdry}), which is not to be expected in the general case.
Some proofs in this subsection  follow  classical lines of reasoning, used also for scalar periodic ordinary differential equations (see, for example, \cite{Ort}).   However,  unlike the ODEs case, it is worth to stress that the lack of  uniqueness of solutions in the past  for  \eqref{eq:general}   can lead to asymptotically periodic motions which  attain the periodic regime in finite time (as it happens in Examples \ref{ex:dry} and \ref{ex:strib}).

Firstly, we observe that  our structural assumptions guarantee  global forward existence and boundedness of the solutions.  Indeed, let us set
\begin{align*}
    v_-:=-R-\norm{\ell_d^-}_{L^1(0,T)}\,, && v_+:=R+\norm{\ell_d^+}_{L^1(0,T)}\,.
\end{align*}
From the bounds in \ref{cond:S2} we easily deduce the following statement:
\begin{prop}
	Let $G$ satisfy the structural assumptions. Then the solution of the Cauchy problem $v(t_0)=v_0$ is bounded between $\min\{v_0-\norm{\ell_d^-}_{L^1(0,T)},v_-\}$ and $\max\{v_0+\norm{\ell_d^+}_{L^1(0,T)},v_+\}$ for every $t\geq t_0$. In particular, this implies global forward existence of solutions.
\end{prop}

Since we have global forward existence and uniqueness of solution, we can define    the Poincaré  map, $\Phi_T\colon \R\to \R$,   that associates to every initial datum $x_0$ at time $t=0$ the value $\Phi_T(x_0)=x(T;x_0)$ of the solution at time $t=T$ of the corresponding Cauchy problem for \eqref{eq:general}. Existence and forward uniqueness of solutions also imply that this map is monotone, while \ref{cond:S1} implies continuous dependence of solutions.   However, unlike for the ODEs case, $\Phi_T$ may fail to be injective.   Nevertheless,  we can  translate many properties of the dynamics    of \eqref{eq:general}  into the discrete dynamics associated to $\Phi_T$,  given by the difference equation $v_{n+1}=\Phi_T(v_n)=\Phi_T^n(v_0)$.  For example, a fixed point $v_*$   of $\Phi_T$  corresponds  to a  $T$ periodic solution $v(t,v_*)$  of  \eqref{eq:general}. 

The next result establishes the asymptotic behaviour of  $\Phi_T.$

\begin{theorem} \label{th:Abs_gen}
	Suppose that $G$ satisfies the structural assumptions.  Then the interval
	\begin{equation}\label{eq:defK}
		K:=[\alpha,\beta]=\bigcap_{i\in \N}\Phi^i_T([v_-,v_+])
	\end{equation}
	is a global attractor for the discrete dynamics induced by $\Phi_T$.  Moreover,  $\Phi_T(\alpha)=\alpha$,  $\Phi_T(\beta)=\beta$  and for any $v_0\in\R,\,   \lim_{i\to +\infty}\Phi^i_T(v_0)=v^*= \Phi_T(v^*)\in [\alpha,\beta]$ 
\end{theorem}
\begin{proof}
	First of all, we observe that, since we have forward uniqueness of solution and our dynamics is in dimension one, then the map $\Phi_T$ is monotone increasing, else two orbits would cross each other. Since by \ref{cond:S2} we have $\Phi_T(v_-)> v_-$ and  $\Phi_T(v_+)< v_+$, using the monotonicity of $\Phi_T$ we deduce that the sequences $\Phi_T^i(v_-)$ and $\Phi_T^i(v_+)$ are, respectively, increasing and decreasing. Thus there exist the limits 
	\begin{align} \label{eq:alphabetadef}
		\alpha:=\lim_{i\to+\infty} \Phi_T^i(v_-)\,, &&\beta:=\lim_{i\to+\infty} \Phi_T^i(v_+)\,.
	\end{align}
	Notice that  the set $[v_-,v_+]$ is forward invariant for the map $\Phi_T$ and that, by monotonicity, $v_-<\alpha\leq \beta <v_+$. Therefore, the right equality in \eqref{eq:defK} is true. By monotonicity we also deduce that $K$ is an attractor for all the orbits starting in $[v_-,v_+]$. Hence it remains to show that every other orbit enters the forward invariant set $[v_-,v_+]$. Let us consider the solution of a general Cauchy problem $v(t_0)=v_0$. We discuss the case $v_0>v_+$;  if $v_0<v_-$ the argument is analogous.
	Suppose by contradiction that $v(t_0+kT)>v_+$ for every $k\in\N$. Hence, by construction, $v(t)>R$ for every $t\in[t_0,+\infty)$. By \ref{cond:S2} we have  $\dot v(t)\leq  \ell_d^+(t)$ 
	for every  $t\geq t_0$. Then, for every $k\in \N$ it holds
	\begin{equation*}
		v(t_0+kT)=v_0+\int\displaylimits_{t_0}^{t_0+kT} \dot v(s)\dd s\leq v_0+k\Lambda \qquad\text{where}\qquad \Lambda= \int\displaylimits_{0}^T\ell_d^+(s)\dd s<0 \,.
	\end{equation*}
		Taking any integer $k>(v_0-v_+)/\abs{\Lambda},$  we get   $v(t_0+kT)<v_+$, and  we arrive to a  contradiction. Hence each orbit reaches the forward invariant set $[v_-,v_+]$ in a finite time (depending on the orbit), therefore $K$ is a global attractor. 
	
	Since $\Phi_T$ is continuous, then $\Phi_T(\alpha)=\alpha$ and $\Phi_T(\beta)=\beta$ follow from \eqref{eq:alphabetadef}. Moreover, since all the orbits are monotone and belong eventually to $K, $ by the continuity of $\Phi_T$  we get immediately  the last part of the statement.   
\end{proof}

 Since  a  point $v_0$  whose orbit   $\{\Phi_T^n(v_0)\}_{n\in\N}$  converges to a fixed point of $\Phi_T$ corresponds to an asymptotically periodic solution of  \eqref{eq:general},   we get immediately the following:

 \begin{theorem} \label{th:asymptper}
Suppose that $G$ satisfies the structural assumptions. Then, for any   solution $v$ of  \eqref{eq:general} there exists a $T$-periodic solution $v^*(t)$  of  \eqref{eq:general} such that $v^*(0)\in K$   and $\lim_{t\to +\infty}  (v(t)-v^*(t))=0$.
\end{theorem}

In particular, we get: 
\begin{corol}\label{cor:Abs_unique}
		Suppose that $G$ satisfies the structural assumptions.
  If \eqref{eq:general} admits only one $T$-periodic solution $v^*$, then $v^*$  is a global attractor for the dynamics in the following sense: given $v_0$,  the solution $v(t;v_0)$ of \eqref{eq:general} such that  $v(t_0;v_0)=v_0$ satisfies
	\begin{equation}\label{eq:globattr}
		\lim_{t\to +\infty} (v(t; v_0)-v^*(t))=0 \,.
	\end{equation}
\end{corol}

A monotonicity assumption on $G$ allows us to better characterize the attactor $K$.
\begin{theorem} \label{th:Abs_weakmon}
 Suppose that $G$ satisfies the structural assumptions and  is   monotone decreasing in $v$. Then, in addition to the conclusions of Theorem \ref{th:Abs_gen}, the set $K$ is the (non-empty) set of all the fixed points of $\Phi_T$, each corresponding to a $T$-periodic solution of \eqref{eq:general}. Furthermore, denoting with $v_\alpha(t)$ the $T$-periodic solution with $v_\alpha(t_0)=\alpha$, all $T$-periodic solutions of \eqref{eq:general} are given by $v_\alpha(t)+(\gamma-\alpha)$ with $\gamma\in [\alpha,\beta]$.
\end{theorem}
\begin{proof}
	Let us consider two solutions $u(t),v(t)$ of \eqref{eq:general}. We observe that, by monotonicity,  for almost every $t$ we have
	\begin{equation}\label{eq:monoton}
		(\dot u(t)  -\dot v(t) )(u(t) -v(t) ) \leq 0 \,.
	\end{equation}
Assume in addition that  $u(t^*)=u^*$ and $v(t^*)=v^*$, with $u^*\geq v^*$.
	Then, for every $t>t^*$ it holds
	\begin{equation}\label{eq:solcompare}
		\abs{u(t)-v(t)}=u(t)-v(t)=u(t^*)-v(t^*)+\int_{t^*}^{t}(\dot u(s)-\dot v(s))\dd s\leq u(t^*)-v(t^*)=\abs{u(t^*)-v(t^*)}
	\end{equation}
where the integral term is non-positive by $\eqref{eq:monoton}$, since  $u-v\geq 0$ due to the forward uniqueness of solutions.
If  $\alpha=\beta$ the proof is completed by noticing that, by Theorem \ref{th:Abs_gen}, all the fixed points of $\Phi_T$ must be in $K$. Suppose now that $\alpha\neq \beta$ and denote with $v_\gamma$ the solution of the Cauchy problem $v(t_0)=\gamma$ for some $\gamma \in K$. We show that $v_\gamma-v_\alpha$ is constant, and therefore $v_\gamma$ is $T$-periodic. Let  $v_\beta$ be the solution of the Cauchy problem $v(t_0)=\beta$.  Being  $v_\beta$ and $v_\alpha$  $T$-periodic, we deduce that $v_\beta=v_\alpha+\beta-\alpha$, since $v_\beta-v_\alpha$ is nonincreasing by \eqref{eq:solcompare} and therefore must be constant. By \eqref{eq:solcompare} we also obtain  that $v_\gamma-v_\alpha$ and $v_\beta -v_\gamma$ are both nonincreasing. Since $v_\beta-v_\gamma=v_\alpha-v_\gamma+\beta-\alpha$ we deduce that $v_\gamma-v_\alpha$ is both nondecreasing and nonincreasing, and therefore it is constant and equal to $\gamma-\alpha$.
\end{proof}

Example~\ref{ex:strib} will illustrate a case where $G$ satisfies the assumptions of Theorem~\ref{th:Abs_gen} but is not monotone. In this case the set $K$ is a nontrivial interval but the dynamics has exactly three periodic solutions.

We now investigate sufficient conditions for the uniqueness of the periodic solution. The first one is strict monotonicity.

\begin{theorem}\label{th:Abs_strmon}
		Suppose that $G$ satisfies the structural assumptions  and  is strictly monotone decreasing in $v$.  Then  \eqref{eq:general} admits exactly one $T$-periodic solution $v^*$, which is a global attractor for the dynamics in the sense of \eqref{eq:globattr}.
\end{theorem}
\begin{proof}
	First of all, let us notice that Theorem~\ref{th:Abs_weakmon} applies, so that, by Corollary~\ref{cor:Abs_unique} we only need to verify that $\beta=\alpha$. Let us denote by $v_\alpha$ and $v_\beta$ their corresponding solutions. By Theorem \ref{th:Abs_weakmon} we have $\dot v_\beta=\dot v_\alpha$ almost everywhere, which by strict monotonicity is possible if and only if $v_\alpha=v_\beta$.
\end{proof}

In general, however, uniqueness of the periodic solution might be achieved also when $G$ is only monotone, provided  that it satisfies some additional structural assumption. We present now a first result of this type in Theorem~\ref{th:Abs_monreg1}; a different structure  with set-valued maps will be considered in Theorem~\ref{th:Dmonot_reg}.

\begin{theorem} \label{th:Abs_monreg1} 
		Let $G\colon\R^2\to\R$  be a single-valued, continuous function, monotone decreasing in the second variable $u$, satisfying \ref{cond:S2} and, moreover, suppose that for every $t$ there exists an unique $u_t$ such that $G(t,u_t)=0$. Then  \eqref{eq:general} admits exactly one $T$-periodic solution $v^*$, which is a global attractor for the dynamics in the sense of \eqref{eq:globattr}.
\end{theorem}	
\begin{proof}
	Noticing that the assumptions of Theorem~\ref{th:Abs_weakmon} are satisfied, let us denote by $v_\alpha$ and $v_\beta$ the solutions with initial values $\alpha$ and $\beta$. By Corollary~\ref{cor:Abs_unique} we only need to verify that $\beta=\alpha$. We observe that $v_\alpha$ is periodic and continuously differentiable, hence there exists $t^*\in[0,T]$ such that $\dot v_\alpha(t^*)=0$. By Theorem \ref{th:Abs_weakmon} we have $\dot v_\beta(t^*)=\dot v_\alpha(t^*)=0$, which by the assumptions of the Theorem is possible only if $ v_\beta(t^*)=v_\alpha(t^*)$, implying $\beta=\alpha$ since we have forward uniqueness of solution.
\end{proof}

\subsection{A lemma}
This technical result will be  a fundamental tool  to prove the existence of a global limit cycle in the case of dry friction (see Theorem \ref{th:Dmonot_reg}).

\begin{defin}\label{def:gamma}
	Given $n$ functions $\alpha_1,\dots \alpha_n\colon \R\to \R$, for every index $1\leq j\leq n$ we denote with $\Gamma^\alpha_j(t)$ the unique value such that
	\begin{itemize}
		\item $\Gamma^\alpha_j(t)\geq \alpha_i(t)$ for at least $j$ distinct indices $i\in\{1,\dots n\}$;
		\item $\Gamma^\alpha_j(t)\leq \alpha_k(t)$ for at least $n-j+1$ distinct indices $k\in\{1,\dots n\}$.
	\end{itemize}
	That is, $\Gamma^\alpha_j(t)$ is the $j$-th smallest value (counting multiplicity) at time $t$ of the collection of functions.
\end{defin}
\begin{lemma}\label{lemma:gamma}
	Let $\alpha_1(t),\dots \alpha_n(t)\colon\R\to \R$ be continuous functions, $T$-periodic, and such that
	\begin{equation*}
		\int_{0}^T\alpha_i(t)\dd t=0 \quad\text{for every $i$.} 
	\end{equation*}	
	Then the functions $\Gamma^\alpha_j(t)$ are continuous and $T$-periodic. Moreover, for every $\delta>0$ and index $j\in\{1,\dots n-1\}$ there exists a measurable set $U_{j,\delta}\subseteq[0,T)$ with positive measure such that $\Gamma^\alpha_{j+1}(t)-\Gamma^\alpha_j(t)<\delta$ for every $t\in U_{j,\delta}$.
\end{lemma}
\begin{proof}
	Continuity and $T$-periodicity of the functions $\Gamma^\alpha_j$ follow from the corresponding property of the functions $\alpha_j$. To verify the last property, suppose by contradiction that there exist an index $k$ and a $\delta>0$ such that $\Gamma^\alpha_{k+1}(t)-\Gamma^\alpha_k(t)\geq\delta$ at almost every $t$. By the continuity of the $\Gamma^\alpha_j$ functions, we deduce that $\Gamma^\alpha_{k+1}(t)>\Gamma^\alpha_k(t)$ for every $t\in [0,T)$. Thus there exists two non-empty set of indexes $J_k^-$ and $J_k^+$ such that $\alpha_i\leq \Gamma^\alpha_{k}$ for every $i\in J_k^-$, while $\alpha_j\geq \Gamma^\alpha_{k+1}$ for every $j\in J_k^+$. Taking any $i\in J_k^-$ and $j\in J_k^+$ we get the contradiction
	\begin{equation*}
		0=\int_0^T\alpha_i(t)\dd t<\int_0^T\alpha_j(t)\dd t=0 \,.
	\end{equation*}
\end{proof}


\section{Discrete models of crawler} \label{sec:discr}
\begin{figure}
    \centering
    	\begin{tikzpicture}[line cap=round,line join=round,x=4mm,y=4mm, line width=1pt, scale=0.9]
		\clip(-1,-2) rectangle (31,7); 
		\draw [line width=1pt, fill=gray!40] (0,0.5)-- (3,0.5)--(3,3.5)-- (0,3.5)-- (0,0.5);
		\draw (3,2)-- (5.2,2.);
		\draw (6.8,2)-- (9.,2.);
		\draw (5.2,2)--(6.3,2);
		\draw (6.5,2)--(6.8,2);
		\draw (6.3,1.6)-- (5.5,1.6)--(5.5,2.4)--(6.3,2.4);
		\draw (5.7,1.8)-- (6.5,1.8)--(6.5,2.2)--(5.7,2.2);
		\draw [line width=1pt, fill=gray!40] (9,3.5)-- (9,0.5)-- (12,0.5)-- (12,3.5)-- (9,3.5);
		\draw (12,2)-- (14.2,2.);
		\draw (15.8,2)-- (18.,2.);
		\draw (14.2,2)--(15.3,2);
		\draw (15.5,2)--(15.8,2);
		\draw (15.3,1.6)-- (14.5,1.6)--(14.5,2.4)--(15.3,2.4);
		\draw (14.7,1.8)-- (15.5,1.8)--(15.5,2.2)--(14.7,2.2);
		\draw [line width=1pt, fill=gray!40] (18,3.5)-- (18,0.5)-- (21,0.5)-- (21,3.5)-- (18,3.5);
		\draw (21,2)-- (23.2,2.);
		\draw (24.8,2)-- (27.,2.);
		\draw (23.2,2)--(24.3,2);
		\draw (24.5,2)--(24.8,2);
		\draw (24.3,1.6)-- (23.5,1.6)--(23.5,2.4)--(24.3,2.4);
		\draw (23.7,1.8)-- (24.5,1.8)--(24.5,2.2)--(23.7,2.2);
		\draw [line width=1pt, fill=gray!40] (27,3.5)-- (27,0.5)-- (30,0.5)-- (30,3.5)-- (27,3.5);
		\draw [] (-2,0)-- (32,0);	
		\draw [] (0,0.5)-- (-0.2,0);
		\draw [] (0.5,0.5)-- (0.3,0);
		\draw [] (1,0.5)-- (0.8,0);
		\draw [] (1.5,0.5)-- (1.3,0);
		\draw [] (2,0.5)-- (1.8,0);
		\draw [] (2.5,0.5)-- (2.3,0);
		\draw [] (3,0.5)-- (2.8,0);
		\draw [] (9,0.5)-- (8.5,0);
		\draw [] (9.5,0.5)-- (9.,0);
		\draw [] (10,0.5)-- (9.5,0);
		\draw [] (10.5,0.5)-- (10.,0);
		\draw [] (11,0.5)-- (10.5,0);
		\draw [] (11.5,0.5)-- (11.,0);
		\draw [] (12,0.5)-- (11.5,0);
		\draw [] (18,0.5)-- (17.8,0);
		\draw [] (18.5,0.5)-- (18.3,0);
		\draw [] (19,0.5)-- (18.8,0);
		\draw [] (19.5,0.5)-- (19.3,0);
		\draw [] (20,0.5)-- (19.8,0);
		\draw [] (20.5,0.5)-- (20.3,0);
		\draw [] (21,0.5)-- (20.8,0);
		\draw [] (27,0.5)-- (26.6,0);
		\draw [] (27.5,0.5)-- (27.1,0);
		\draw [] (28,0.5)-- (27.6,0);
		\draw [] (28.5,0.5)-- (28.1,0);
		\draw [] (29,0.5)-- (28.6,0);
		\draw [] (29.5,0.5)-- (29.1,0);
		\draw [] (30,0.5)-- (29.6,0);
		\draw (1.5,2) node[] {$m_1$};
		\draw (10.5,2) node[] {$m_2$};
		\draw (19.5,2) node[] {$m_3$};
		\draw (28.5,2) node[] {$m_4$};
		\draw (1.5,4.5) node[] {$x_1(t)$};
		\draw (10.5,4.5) node[] {$x_2(t)$};
		\draw (19.5,4.5) node[] {$x_3(t)$};
		\draw (28.5,4.5) node[] {$x_4(t)$};
		\fill [pattern = north east lines] (-2,0) rectangle (32,-0.6);
	\end{tikzpicture}
    \caption{The discrete model of crawler studied in Section~\ref{sec:discr}.}
    \label{fig:discr}
\end{figure}
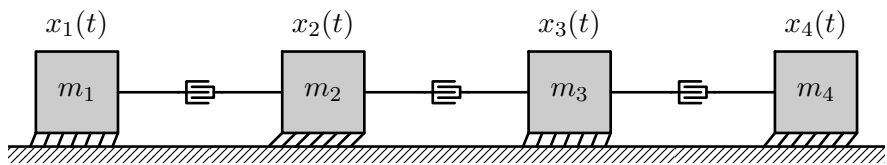
We consider a system composed of $n$ linked point masses $m_1,\dots m_n$ on a line, illustrated in Figure~\ref{fig:discr}. Similar models have been considered, for instance, in \cite{BSZZB17,BSZZ18,FigKny,GidRiv,Marvi,WagLau}, and in \cite{Gid18, ColGidVil} at the quasistatic regime. We assume that each mass is non-negative, but we require the total mass of the system $M:=\sum m_i$ to be positive. We assume that all masses are placed on a line, and denote with the vector $x(t)=(x_1(t),x_2(t),\dots,x_n(t))\in\R^{n}$ their position at any time $t$.  
The shape of the locomotor, namely the relative distances between the position of the masses, are assigned. This means that the vector $x$ can be decomposed in an unknown scalar value indicating the overall position of the crawler and a prescribed $(n-1)$-dimensional vector describing its shape. Several choices of the reference point are possible, e.g., the barycentre, the average position, the head, the tail; clearly this does not affect the qualitative behaviour of the system. 
In the following, we identify the position of the locomotor with its barycentre $\bar x(t)$ and denote with $z_i(t)$ the position of the $i$-th mass relative to the barycentre, namely
\begin{align*}
	\bar x(t)=\frac{1}{M}\sum_{i=1}^n m_ix_i(t) \,, && z_i(t):=x_i(t)-\bar x(t) \,.
\end{align*}

Notice that that the vector $z(t)\in\R^n$ is actually contained in the $(n-1)$-dimensional subspace defined by the constraint $\sum m_i z_i=0$. Moreover, assigning an input $z(t)$ is equivalent to controlling the (signed) distance between consecutive masses, since $x_i-x_j=z_i-z_j$.

In what follows, unless stated differently, we will always  assume that the following holds for the shape functions:
\begin{enumerate}[label=\textup{(D\arabic*)}]
\item \label{cond:Dzreg} the functions  $z_i(t)$ are $T$-periodic and Lipschitz continuous.
\end{enumerate}

Each mass is subject to a time-dependent friction force $F_i(t,\dot x_i)$. In concrete situation, time-dependence can be produced by various mechanisms: a change in the normal force on the contact surfaces (e.g.~due to an expansion while crawling in a tube) or in their properties  (e.g.~a change in the tilt angle of bristles or scales on the body of the crawler \cite{GidDeS18,Marvi}) or more complex phenomena (cf.~\cite{Rehor}).   We assume that each friction force $F_i$ is $T$-periodic in time; the possible types of friction force-velocity law will be discussed below. Notice that, in general, the functions $F_i$ may be set-valued, in order to account for discontinuous forces such as dry friction, for which a yield force has to be reached in order to slide. This implies that, in general, the dynamics of the system is described by a differential inclusion. 

Finally, we also allow an external force $B(t)$ acting on the system. Since we are in a one-dimensional setting and the shape of the locomotor is predetermined, the force might be considered as applied to the barycentre. We assume
\begin{enumerate}[label=\textup{(D\arabic*)},resume]
	\item \label{cond:DBreg} the function  $B(t)$ is $T$-periodic, bounded and measurable.
\end{enumerate}
We point out that locomotion is, \emph{per se}, driven by an internal actuation, in our case represented by shape change $z_i$, so in our examples we will focus on $B\equiv 0$. External forces may however be relevant to describe additional effects: for instance, crawling on a slope with angle $\theta$ corresponds to $B\equiv -Mg\sin\theta$, where $g$ denotes gravity.

Since the shape of the locomotor is prescribed, to study its evolution we just have to consider that of its barycenter $\bar x(t)$, that is determined by the sum of the external forces $B$ and $F_i$.
Denoting $v:=\dot{\bar x}$ and $w:=\dot z$, we can rewrite the dynamics of the barycenter as
\begin{equation}\label{eq:gdiscr}
	\dot v\in G(t,v):=\frac{1}{M}\left(B(t)+\sum_{i=0}^n  F_i(t,v+w_i(t))\right) \,.
\end{equation}

Notice that $G(t,v)$ is defined up to a set of zero measure (which corresponds to the points of non differentiability of the given shape functions $z_i$). However,  we can redefine arbitrarily $G$ on such set without changing the set of solutions.  Also, note that the  functions $x_i=\bar x+z_i$  are the sum of  $\bar x\in W^{2,1}([0,T],\R)$  and $z\in W^{1,\infty}([0,T],\R)$, so that, in general,
the evolution of the material points $x_i$ is less regular. This corresponds to the the exchange of impulsive internal forces at the times at which $z_i$ is not differentiable, so that the dynamics of a single material point can be expressed only in a distributional sense. 

We will discuss several classes of friction forces, providing increasingly stronger convergence properties of the dynamics. We look for limit cycles of Equation \eqref{eq:gdiscr}: indeed, a periodic solution $\bar v$ of \eqref{eq:gdiscr} corresponds to a relative-periodic evolution of the state $x$, with geometric phase $\gamma=\int_0^T\bar v(t)\dd t$.

\paragraph{Generic friction force (possibly non-monotone)}
We begin by discussing the situation of a generic friction force.

\begin{enumerate}[label=\textup{(D\arabic*)},resume]
	\item \label{cond:Dreg} For every index $i$, the function $F_i$ 
	is a set-valued map $T$-periodic in $t$ and locally bounded by a measurable function.
	Moreover, it can be decomposed as $F_i(t,u)=A_i(t,u)+\psi_i(t,u)$, where $A_i$ is maximal monotone decreasing  in $u$ for every $t$ and measurable in $t$ for every $u$, while $\psi_i$ is a Carathéodory function locally Lipschitz continuous in the variable $u$ uniformly in $t$.
	\item \label{cond:Dcoer} There exists $R>0$ and, for every index $i$, two 
	$T$-periodic measurable functions $\ell_i^\pm(t)\colon\R\to\R$ such that, for almost every $t\in \R$ 
	\begin{align*}
	F_i(t,u)&\geq -\ell_i^-(t) &&\text{for  every  $u\leq -R$,}\\
	F_i(t,u)&\leq \ell_i^+(t) &&\text{for  every  $u\geq R$,}
	\end{align*} 
with
	\begin{equation*} \int_{0}^T \left(B(s)+\sum_{i=1}^n\ell_i^+(s)\right)\dd s<0<\int_{0}^T \left(B(s)-\sum_{i=1}^n\ell_i^-(s)\right)\dd s \,.
	\end{equation*} 

\end{enumerate}
The regularity condition \ref{cond:Dreg} allows to include $F_i(t,\cdot)$  that are, in general, non-monotone, hence accounting for a Stribeck effect \cite{AdlGoe,OdeMar}. Notice, however, that the left limit of $F_i(t,\cdot)$ at discontinuity points is always greater than the right one, due to the monotonicity of $A_i$.

In concrete cases, we expect the forces $F_i$ to act always against the direction of motion. We would like, however, to include also situations in which one point-mass is not in contact with the surface during a certain phase of the period, for instance due to lifting or shrinking of the corresponding part of the body, resulting in a temporarily null $F_i$. On the other hand, this should be a constrained phenomenon, avoiding a perpetual frictionless sliding in one direction. This is, intuitively, the meaning of condition \ref{cond:Dcoer} for $B\equiv 0$, although, for mathematical generality, such property is required only  at sufficiently large speeds.
For $B\neq 0$, \ref{cond:Dcoer} also guarantees that the external load does not overcome friction forces, i.e., the locomotion component of the dynamics. For instance, in the case of crawling on a slope, it guarantees that an inactive crawler would not slide downward accelerating endlessly.

\begin{theorem}\label{th:Dgen}
	Let us consider a discrete model of crawler as above, satisfying \ref{cond:Dzreg}, \ref{cond:DBreg}, \ref{cond:Dreg} and \ref{cond:Dcoer}. 
 Then all the solutions $v$ of \eqref{eq:gdiscr} are asymptotically $T$-periodic. 
 Moreover, there exists two, possibly identical, $T$-periodic solutions $v_\alpha,v_\beta$ of the dynamics such that	every solution $v$ of \eqref{eq:gdiscr} satisfies
	\begin{equation}
	\lim_{t\to +\infty}\dist (v(t),[v_\alpha(t),v_\beta(t)])=0 \,.
	\end{equation}
\end{theorem}
\begin{proof}
	We observe that the function $G$ in \eqref{eq:gdiscr} is, by \ref{cond:Dreg}, of the form \eqref{eq:Gform}; furthermore, by  \ref{cond:Dcoer} and the Lipschitz continuity of the $z_i$, it satisfies \ref{cond:S2}.  Indeed, denoting by $\Lambda_z$ the Lipschitz constant of the $z_i$, and taking $R:=r+\Lambda_z+1$ we get, for every $v\geq R$ and $t\in \R$,
	\begin{align*}
		G(t,v)=\frac{1}{M}\left(B(t)+\sum_{i=0}^n  F_i(t,v+w_i(t))\right)\leq \frac{1}{M}\left(B(t)+\sum_{i=0}^n  \ell_i^+(t)\right)=:\ell_d^+(t)
	\end{align*} 
	where the function $\ell_d^+(t)$ thus defined satisfies, by \ref{cond:Dcoer}, the requirements of  \ref{cond:S2}. The other bound in \ref{cond:S2} can be deduced analogously. We can therefore apply  Theorem~\ref{th:asymptper}
 and conclude. 
\end{proof}

\paragraph{Monotone friction forces.}
Most of the classical examples of friction forces, for instance viscous of dry friction, are however monotone with respect to velocity. This means that in \ref{cond:Dreg} we can set $\psi_i\equiv 0$; namely
\begin{enumerate}[label=\textup{(D\arabic*)},resume]
	\item \label{cond:Dmono} For every index $i$, the function $F_i\colon\R\times \R\to \PP(\R)\setminus\emptyset$ is a set-valued map, locally bounded by a measurable function, $T$-periodic in $t$, maximal monotone decreasing  in $u$ for every $t$ and measurable in $t$ for every $u$.
\end{enumerate}
 The most important example   of this family, excluding strictly monotone ones, is that of dry friction, defined by $F_i(t,u)=-\alpha\, \partial_u \abs{u}$.

\begin{theorem}\label{th:Dmonot}
	Let us consider a discrete model of crawler as above, satisfying \ref{cond:Dzreg}, \ref{cond:DBreg}, \ref{cond:Dcoer} and \ref{cond:Dmono}. Then, for every solution $v(t)$ of \eqref{eq:gdiscr} there exists a $T$-periodic solution $\bar v$ such that
	\begin{equation*}
		\lim_{t\to+\infty} (v(t)-\bar v(t))=0
	\end{equation*} 
 Moreover, there exists an  interval $[a,b]$,   possibly degenerate,  such that $\bar v$ is a $T$-periodic solution if and only if $\bar v=v^*+c$ for some $c\in[a,b]$. 
\end{theorem}
\begin{proof}
	The function $G$ in \eqref{eq:gdiscr} is, by \ref{cond:Dmono}, of the form \eqref{eq:Gform} and monotone decreasing in $v$; furthermore, by \ref{cond:Dcoer} and the Lipschitz continuity of the $z_i$, it satisfies \ref{cond:S2}. The conclusion follows by Theorem~\ref{th:Abs_weakmon}.
\end{proof}

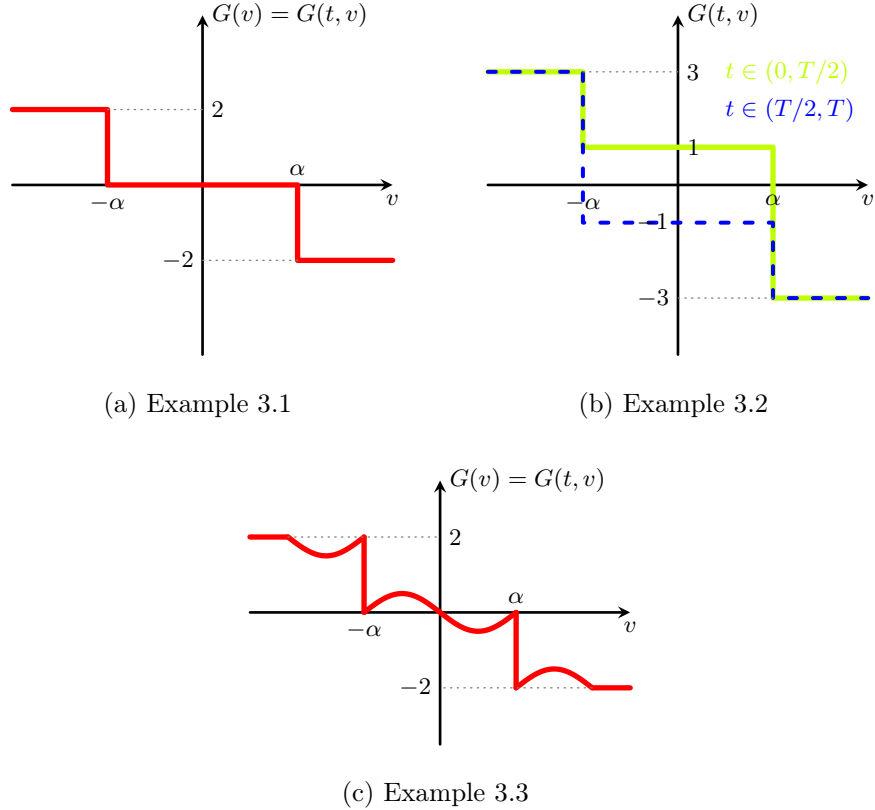
\begin{figure}[t]
    \centering
\subcaptionbox{Example~\ref{ex:dry}\label{fig:ex1}}
{	\begin{tikzpicture}[line cap=round,line join=round,>=stealth,x=1cm,y=1cm, line width=1,scale=0.5]
\clip(-6,-5) rectangle (6,5);
\footnotesize
\draw[->] (-5,0)--(5,0)node[below]{$v$};
\draw[->] (0,-4.5) --(0,4.5) node[right]{$G(v)=G(t,v)$};
\draw [line width=2pt,color=red] (-5,2)-- (-2.5,2)--(-2.5,0)--(2.5,0)--(2.5,-2)--(5,-2);
\draw [line width=0.5pt,dotted, gray] (-2.5,2)--(0,2) node[anchor=west, black]{$2$};
\draw [line width=0.5pt,dotted,gray] (2.5,-2) -- (0,-2) node[anchor=east, black]{$-2$};
\draw (2.5,0) node[anchor=south] {$\alpha$};
\draw (-2.5,0) node[anchor=north] {$-\alpha$};
\end{tikzpicture}}
\subcaptionbox{Example~\ref{ex:drystar}\label{fig:ex2}}
{	\begin{tikzpicture}[line cap=round,line join=round,>=stealth,x=1cm,y=1cm, line width=1,scale=0.5]
\clip(-6,-5) rectangle (6,5);
\footnotesize
\draw[->] (-5,0)--(5,0)node[below]{$v$};
\draw[->] (0,-4.5) --(0,4.5) node[right]{$G(t,v)$};
\draw [line width=2pt,color=lime] (-5,3)-- (-2.5,3)--(-2.5,1)--(2.5,1)--(2.5,-3)--(5,-3);
\draw [line width=1.5pt,color=blue,loosely dashed] (-5,3)-- (-2.5,3)--(-2.5,-1)--(2.5,-1)--(2.5,-3)--(5,-3);
\draw [line width=0.5pt,dotted, gray] (-2.5,3)--(0,3) node[anchor=west, black]{$3$};
\draw [line width=0.5pt,dotted,gray] (2.5,-3) -- (0,-3) node[anchor=east, black]{$-3$};
\draw (2.5,0) node[anchor=north] {$\alpha$};
\draw (-2.5,0) node[anchor=north] {$-\alpha$};
\draw (0,-1) node[anchor= east] {$-1$};
\draw (0,1) node[anchor=west] {$1$};
\draw [lime] (1,3)node[anchor=west] {$t\in(0,T/2)$};
\draw [blue] (1,2)node[anchor=west] {$t\in(T/2,T)$};
\end{tikzpicture}}
\subcaptionbox{Example~\ref{ex:strib}\label{fig:ex3}}
{	\begin{tikzpicture}[line cap=round,line join=round,>=stealth,x=1cm,y=1cm, line width=1, scale=0.5]
			\clip(-6,-4) rectangle (6,5);
			\footnotesize
			\draw[->] (-5,0)--(5,0)node[below]{$v$};
			\draw[->] (0,-3.5) --(0,3.5) node[right]{$G(v)=G(t,v)$};
			\draw[line width=2pt,color=red,smooth,samples=100,domain=-2:2] plot(\x,{-0.5*sin(90*\x)});
			\draw[line width=2pt,color=red,smooth,samples=100,domain=2:4] plot(\x,{-2-0.5*sin(90*\x)});
			\draw[line width=2pt,color=red,smooth,samples=100,domain=-4:-2] plot(\x,{2-0.5*sin(90*\x)});
			\draw[line width=2pt,color=red] (-5,2)--(-4,2);
			\draw[line width=2pt,color=red] (4,-2)--(5,-2);
			\draw[line width=2pt,color=red] (-2,2)--(-2,0);
			\draw[line width=2pt,color=red] (2,-2)--(2,0);
			\draw [line width=0.5pt,dotted, gray] (-4,2)--(0,2) node[anchor=west, black]{$2$};
			\draw [line width=0.5pt,dotted,gray] (4,-2) -- (0,-2) node[anchor=east, black]{$-2$};
			\draw (2,0) node[anchor=south] {$\alpha$};
			\draw (-2,0) node[anchor=north] {$-\alpha$};
	\end{tikzpicture}}
\caption{The function $G$ in Examples \ref{ex:dry}, \ref{ex:drystar} and \ref{ex:strib}. In Examples \ref{ex:dry} and \ref{ex:strib}, the resulting function $G(t,v)$ does not depend on time $t$. In Example~\ref{ex:drystar} the function $G$ switches every half-period. The behaviour for $t\in(0,T/2)$ is illustrated by the solid green line, whereas   for $t\in(T/2,T)$ it is given by the dashed blue one.}
\label{fig:ex}
\end{figure}

\begin{example} \label{ex:dry}
	Let us consider the following example with $n=2$, $M=1$ and $B\equiv 0$, characterized by isotropic dry friction, constant in time and equal on both point masses, and a $T$-periodic triangular wave actuation on shape; more precisely we set
	\begin{equation*}
		F_1(t,u)=F_2(t,u)=-\partial_u \abs{u} \qquad\qquad w_1(t)=-w_2(t)=\begin{cases}
			-\alpha &\text{for $t\in(0,T/2)$}\\
			\alpha &\text{for $t\in(T/2,T)$}
		\end{cases}
	\end{equation*}
for some $\alpha\neq 0$. The function $G$ from \eqref{eq:gdiscr} is constant in time (except for times of the form $t=kT/2$, $k\in\Z$, where we can extend it by continuity), with values
\begin{equation*}
	G(t,u)=-\partial_u \left(\abs{u-\alpha}+\abs{u+\alpha}\right) 
\end{equation*}
We notice that $G$ is monotone decreasing in $u$, with $0\in G(t,u)$ if and only if $u\in[-\alpha,\alpha]$. Hence the set of periodic solutions provided by Theorem~\ref{th:Dmonot} consists of the constant functions $\bar v\in[-\alpha,\alpha]$.

We emphasize that, in this example, different initial conditions might lead to a different asymptotic velocities of the crawler, even having opposite sign.
\end{example}

In the previous example, all periodic solution have constant velocity $\dot v$ of the barycenter, so we might ask ourselves whether this is a necessary feature for multiplicity. The issue is furthermore  relevant since the friction configuration considered in the example  is problematic \emph{at the quasistatic regime}, violating uniqueness condition $(\ast)$ in \cite{Gid18},  leading to multiplicity of solutions (to the initial value problem) with the same shape change but different velocity \cite[Example~3.2]{Gid18}.
The answer is negative, as we show with the following example satisfying $(\ast)$, which therefore plays no role \emph{in the dynamic regime}. 

\begin{example}\label{ex:drystar}
  Let us modify Example~\ref{ex:dry} just by doubling the friction on the first point mass, namely $F_1(t,u)=-2 \partial_u\abs{u}$, and setting $T<4\alpha$. 
  Considering only the relevant part, we have
  \begin{equation*}
	G(t,u)=\begin{cases}
			1 &\text{for $t\in(0,T/2), u\in(-\alpha,\alpha)$}\\
			-1 &\text{for $t\in(T/2,T), u\in(-\alpha,\alpha)$}
		\end{cases} 
\end{equation*}
Hence, for each $u_0\in[-\alpha, \alpha-T/2]$ we have a parametric set of periodic solutions
\begin{equation}
    u=\begin{cases} u_0 +t &\text{for $t\in(0,T/2)$}\\
    u_0 +(T-t) &\text{for $t\in(T/2,T)$}
    \end{cases}
\end{equation}
of the dynamics, each with average velocity $u_0+T/4$.
\end{example}

\begin{example}\label{ex:strib}
	We now consider a variation of Example~\ref{ex:dry} with the addition of a Stribeck effect on friction. We set
	\begin{equation*}
	F_1(t,u)=F_2(t,u)=-\partial_u \abs{u}+\psi(u) \qquad\qquad \psi(u)=\begin{cases}
		\frac{1}{2}\sin(\frac{\pi}{\alpha} u ) &\text{for $u\in(-\alpha,\alpha)$}\\
		0 &\text{elsewhere}
	\end{cases}
\end{equation*}	
	with $\alpha, w_1, w_2, B$ as in Example~\ref{ex:dry}. Hence
	 \begin{equation*}
		G(t,u)=-\partial_u \left(\abs{u-\alpha}+\abs{u+\alpha}\right) +\psi(u-\alpha)+\psi(u+\alpha) \,.
\end{equation*}

The dynamics \eqref{eq:gdiscr} admits exactly three $T$-periodic solutions: the constant solutions $v^*\equiv-\alpha$ and $v^{**}\equiv\alpha$ are 
semistable, with basins of attraction given  respectively by $(-\infty,-\alpha]$ and $[\alpha,+\infty)$, whereas the constant solution $v^0\equiv0$ is stable with basin of attraction $(-\alpha, \alpha)$. All the solutions of the basins of attractions of $v^*$ and $v^{**}$ attain the corresponding periodic regime in finite time.
\end{example}

\paragraph{Strictly monotone friction forces.}
We strengthen the assumptions of the previous case, requiring strict monotonicity of the functions $F_i(t,\cdot)$.
\begin{enumerate}[label=\textup{(D\arabic*)},resume]
\item \label{cond:Dstrict}  For every index $i$, the function $F_i\colon\R\times \R\to \PP(\R)\setminus\emptyset$ is a set-valued map, locally bounded by a measurable function, $T$-periodic in $t$, maximal strictly monotone decreasing  in $u$ for every $t$ and measurable in $t$ for every $u$.
\end{enumerate}

The most important example in this family is that of viscous friction $F_i(t,u)=-\mu u$.
For set-valued maps, we mention the case of Bingham friction, defined by $F_i(t,u)=-\mu u -\alpha\, \partial_u\abs{u}$.

\begin{theorem}\label{th:Dstrict}
	Let us consider a discrete model of crawler as above, satisfying \ref{cond:Dzreg}, \ref{cond:DBreg}, \ref{cond:Dcoer}  and \ref{cond:Dstrict}. Then the dynamics \eqref{eq:gdiscr} admits exactly one $T$-periodic solution  $v^*$, which is a global attractor for the dynamics, i.e.~every solution $v(t)$ of \eqref{eq:gdiscr} satisfies
	\begin{equation}\label{eq:Dglobattr}
		\lim_{t\to+\infty}  (v(t)- v^*(t))=0 \,.
	\end{equation} 
\end{theorem}
\begin{proof}
	The function $G$ in \eqref{eq:gdiscr} is, by \ref{cond:Dstrict}, of the form \eqref{eq:Gform} and strictly monotone decreasing in $v$; furthermore, by  \ref{cond:Dcoer} and the Lipschitz continuity of the $z_i$, it satisfies \ref{cond:S2}. The conclusion follows by Theorem~\ref{th:Abs_strmon}.
\end{proof}

\begin{example}[Asymptotically incompetent crawler]\label{ex:incomp}
Let us consider the case with $n\geq 2$, $B\equiv 0$ and (constant, isotropic) viscous friction, namely $F_i(t,u)=-\mu_i u$, with $\mu_i>0$. The unique $T$-periodic solution $v^*$ of \eqref{eq:gdiscr} satisfies
\begin{equation*}
0=\int_0^T \dot v^*(t)\dd t=-\sum_{i=1}^n\frac{\mu_i}{M}\int_0^T  v^*(t)\dd t -\sum_{i=1}^n\frac{\mu_i}{M}\int_0^Tw_i(t)\dd t \,.
\end{equation*}
Since the $w_i$ are the derivatives of the $T$-periodic functions $z_i$, we deduce that $\int_0^Tw_i(t)\dd t=0$ and therefore
\begin{equation*}
\int_0^T  v^*(t)\dd t =0 \,.
\end{equation*}
The fact that, in discrete models, constant, isotropic viscous friction leads to an ``incompetent'' crawler, with a null net displacement, was observed, in a quasi-static setting, in \cite{DeSTat12}. Here we see that, considering also inertia, such a crawler is still incompetent, but in an asymptotic sense. 
\end{example}

The issue in Example~\ref{ex:incomp} can be overcome still in a viscous setting, if we consider anisotropic viscosity or time-dependent viscosity. We illustrate this second possibility with the following example, that can be seen a viscous adaptation of two-anchor crawling.
\begin{example}
\label{ex:comp}
Let us set $n=2$, $M=2$, $B\equiv 0$ and
\begin{equation*}
F_1(t,u)=-(2+\sin t)u \qquad F_2(t,u)=-(2-\sin t)u \qquad z_1(t)=-z_2(t)=2+\cos t
\end{equation*}
so that the dynamics \eqref{eq:gdiscr} reads
\begin{equation*}
	\dot v=-2v+\sin^2 t
\end{equation*}	
	whose solutions have the form
\begin{equation*}
	v(t)=\frac{\sin^2 t}{2}+\frac{\cos 2t-\sin 2t}{8}+c e^{-2t}
\end{equation*}	
	with $c\in\R$. Clearly the solution $v^*$ associated with $c=0$ is the unique $T$-periodic solution of the system and a global attractor for the dynamics. Unlike  the previous case, however, it is easily verified that, in the asymptotic regime, each actuation cycle of duration $2\pi$ produces a forward net displacement of the crawler of $2\pi/4$.
\end{example}

\paragraph{Dry friction forces with additional time-regularity}

\begin{theorem}\label{th:Dmonot_reg}	
	We make the following assumptions:
	\begin{enumerate}[label=\textup{(\roman*)}]
		\item \label{cond:drysmooth}
		Each mass is affected by dry friction with nonzero coefficients continuous in time, i.e.
		\begin{equation*}
		F_i(t,u)= \begin{cases}
			\{-\mu_i^+(t)\} & \text{for $u>0$}\\
			\{\mu_i^-(t)\} & \text{for $u<0$}\\
			[-\mu_i^+(t),\mu_i^-(t)] & \text{for $u=0$}
		\end{cases}
		\end{equation*}
	where the $\mu_i^\pm$ are continuous, $T$-periodic and positive;
	\item the shape functions $z_i(t)$ are continuously differentiable and $T$-periodic;
	\item the external load $B(t)$ is continuous, $T$-periodic and such that
	\begin{equation*}
	\int_0^T \left( B(s)-\sum_{i=1}^n \mu_i^+(s)\right)\dd s <0< \int_0^T \left(B(s)+\sum_{i=1}^n \mu_i^-(s)\right)\dd s 
	\end{equation*}
	\end{enumerate}
Then the dynamics \eqref{eq:gdiscr} admits exactly one $T$-periodic solution  $v^*$, which is a global attractor for the dynamics in the sense of \eqref{eq:Dglobattr}.
\end{theorem}

\begin{proof}
We observe that the assumptions of Theorem~\ref{th:Dmonot} are satisfied, hence there exists at least a $T$-periodic solution $v^*$ and, if the dynamics admits additional $T$-periodic solutions, they are of the form $\bar v(t)=v^*(t)+\delta$ for some constant $\delta$.
To show that $v^*$ is unique, we suppose by contradiction that there exists a second $T$-periodic solution $\bar v(t)=v^*(t)+\delta$, assuming without loss of generality $\delta>0$. 

We notice that  by  \ref{cond:drysmooth}  the function $G$ in \eqref{eq:gdiscr} is such that for every $t$ and for every $u_a,u_b$ satisfying $u_a<-w_i(t)<u_b$  we have  $G(t,u_a)>G(t,u_b)$.

Recalling that $w_i(t)=\dot z_i(t)$, so that $\int_0^T w_i(t)=0$, we consider the functions $\Gamma^{-w}_j$ as in Definition~\ref{def:gamma} with $\alpha_i=-w_i$.  Since $v^*$ is continuous, we must be in one of the following cases.
\begin{itemize}
	\item \emph{there exists a set $U\subset [0,T)$ with positive measure and an index $k\in\{1,\dots n\}$ such that $v^*(t)<\Gamma^{-w}_k(t)<\bar v(t)$ for every $t\in U$.} This implies $G(t,v^*(t))>G(t,\bar v(t))$ for $t\in U$, and therefore $\dot v^*>\dot{\bar v}$ almost everywhere on $U$. On the other hand, by $\bar v(t)=v^*(t)+\delta$ we deduce that $\dot v^*=\dot{\bar v}$ almost everywhere, obtaining a contradiction.
	\item \emph{there exists an index $k\in\{1,\dots n-1\}$ such that $\Gamma^{-w}_k(t)\leq v^*(t)<\bar v=v^*(t)+\delta\leq \Gamma^{-w}_{k+1}(t)$ for every $t$.} This is in contradiction with the final proposition in Lemma~\ref{lemma:gamma}.
	\item \emph{$v^*(t)<\bar v(t)\leq \Gamma^{-w}_1(t)$ for almost every $t$.} Then by \ref{cond:drysmooth} we deduce that $\dot v^*(t)\in G(t,v^*(t))>0$ for almost every $t$, which is a contradiction with $v^*$ being absolutely continuous and $T$-periodic.
	\item \emph{$\bar v(t)>v^*(t)\geq\Gamma^{-w}_n(t)$ for almost every $t$.} Then by \ref{cond:drysmooth} we deduce that $\dot {\bar v}(t)\in G(t,v^*(t))<0$ for almost every $t$, which is a contradiction with $\bar v$ being absolutely continuous and $T$-periodic.
\end{itemize}
Since in each case the existence of a second $T$-periodic solution $\bar v=v^*+\delta$ leads to a contradiction, we deduce that $v^*$ is the unique $T$-periodic solutions of \eqref{eq:gdiscr} and therefore a global attractor for the dynamics in the sense of \eqref{eq:Dglobattr}.
\end{proof}

\section{Continuous models of crawler} \label{sec:cont}
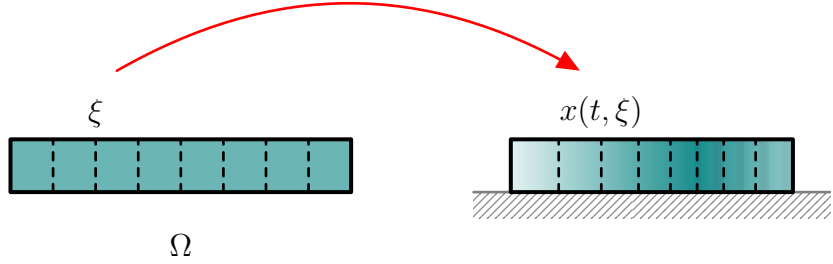
\begin{figure}
    \centering
    	\begin{tikzpicture}[line cap=round,line join=round,>=triangle 45,x=1cm,y=1cm, scale=0.7, line width=1pt]\large
		\clip(-10,-2) rectangle (6,4.5);
		\draw[line width=1pt,gray] (-0.7,-0)--(6,-0);
		\fill [pattern = north east lines, pattern color=gray] (-0.7,0) rectangle (6,-0.5);
		\fill [left color=teal!10!white, right color=teal!90!white] (0,0)--(3.5,0)--(3.5,1)--(0,1)--cycle;
		\fill [left color=teal!90!white, right color=teal!42!white,] (3.5,0)--(5.3,0)--(5.3,1)--(3.5,1)--cycle;
		\draw [line width=1.5pt] (0,0)--(5.3,0)--(5.3,1)--(0,1)--cycle;
		\draw[dashed,] (0.9,0)--(0.9,1);
		\draw[dashed] (1.7,0)--(1.7,1) node[above]{$x(t,\xi)$};
		\draw[dashed] (2.4,0)--(2.4,1);
		\draw[dashed] (3,0)--(3,1);
		\draw[dashed] (3.5,0)--(3.5,1);
		\draw[dashed] (4,0)--(4,1);
		\draw[dashed] (4.6,0)--(4.6,1);
				\draw [fill=teal!58!white,line width=1.5pt] (-3,0)--(-9.4,0)--(-9.4,1)--(-3,1)--cycle;
		\draw[dashed] (-3.8,0)--(-3.8,1);
		\draw[dashed] (-4.6,0)--(-4.6,1);
		\draw[dashed] (-5.4,0)--(-5.4,1);
		\draw[dashed] (-6.2,0)--(-6.2,1);
		\draw[dashed] (-7,0)--(-7,1);
		\draw[dashed] (-7.8,0)--(-7.8,1) node[above]{$\xi$};
		\draw[dashed] (-8.6,0)--(-8.6,1);
		\draw [->,red] (-7.4,2.3) to [out=30,in=150] (1.3,2.3);
		\draw (-6.2,-1) node{$\Omega$};
	\end{tikzpicture}
    \caption{The continuous model of crawler studied in Section~\ref{sec:cont}.}
    \label{fig:cont}
\end{figure}

Let us now consider that the body of our crawler is, in the reference configuration, the interval $\Omega=[\xi_a,\xi_b]$, so that its position  at time $t$  is described by a function $x(t,\xi)\colon [0,T] \times\Omega \to \R$, cf.~Figure~\ref{fig:cont}. Similar models have been considered, for instance, in \cite{BPZZ16}, and, at the quasistatic regime, in \cite{Ago,DeSGidNos,DeSTat12,Gid18}.  We assume that the crawler has a total mass $M>0$ distributed according to a non-negative mass density $\rho\in L^\infty(\Omega,[0,+\infty))$, and define $\bar x(t):=\int_\Omega x(t,\xi)\dd \rho(\xi)$ as the position of the barycentre.   Analogously to the discrete case, the shape of the crawler is prescribed. To do so, we prescribe the deformation gradient $\nabla_\xi x(t,\xi)=\phi(t,\xi)$ where

\begin{enumerate}[label=\textup{(C\arabic*)}]
	\item \label{cond:Czreg} $\phi\in W^{1,\infty}(\R,L^\infty(\Omega,\R))$ is $T$-periodic in $t$; moreover  we assume that there exist $\phi_\mathrm{max},\phi_\mathrm{min}$ such that $0<\phi_\mathrm{min}\leq\phi(t,\xi)\leq \phi_\mathrm{max}$ for almost every $(t,\xi)\in \R\times \Omega$. 
\end{enumerate}	
	As in the discrete case, we we denote the relative position of an element $\hat \xi\in\Omega$ with respect to the barycentre as $z(t,\hat \xi)=x(t,\hat \xi)-\bar x(t)$. We observe that
\begin{equation}
z(t,\hat \xi)=x(t,\hat \xi)-\bar x(t)=\frac{1}{M}\int_\Omega \rho(\xi)(x(t,\hat \xi)-x(t,\xi))\dd \xi=\frac{1}{M}\int_\Omega \rho(\xi)\int_\xi^{\hat \xi}\phi(t,s)\dd s\dd \xi
\end{equation}
	so that $z\in W^{1,\infty}(\R,W^{1,\infty}(\Omega,\R))$. 
	 We also notice that $\int_\Omega z(t,\xi)\dd \rho(\xi)=0$.
	 
	 The crawler is subject to a time-dependent friction force per unit length $f(t,\xi, \dot x_i)$ acting along the body. We assume that $t\mapsto f(t,\xi,\eta)$ is $T$-periodic. The possible types of friction force-velocity law will be discussed below, analogously to the discrete case. 
	 
	 As in the discrete case, we also allow an external force $B(t)$ acting on the system, which, without loss of generality, is considered as applied to the barycentre. We assume
	 \begin{enumerate}[label=\textup{(C\arabic*)},resume]
	 	\item \label{cond:CBreg} the function  $B(t)$ is $T$-periodic, bounded and measurable.
	 \end{enumerate}

The equation of motion reads
\begin{equation}
	M\ddot {\bar x} -B(t) \in\int_\Omega f\bigl(t,\xi,\dot x(t,\xi)\bigr)\dd \xi=
	\int_\Omega f\bigl(t,\xi,\dot{\bar x}(t) +\dot z(t,\xi)\bigr)\dd \xi \,.
\end{equation}

Recalling that $z$ is assigned, so that the problem has to be solved only for $\bar x$, we proceed as in the previous section and set $v(t)=\dot{\bar x}(t)$, $w(t,\xi)=\dot z(t,\xi)$ and obtain
\begin{equation}\label{eq:Cgen}
	\dot v(t) \in G(t,v):=\frac{1}{M}\left( B(t)+\int_\Omega f\bigl(t,\xi,v(t) +w(t,\xi)\bigr)\dd \xi \right) \,.
\end{equation}
As in discrete case, we study limit cycles for \eqref{eq:Cgen}, since its periodic solutions correspond to relative-periodic evolutions of the state $x$. 

\paragraph{Generic friction force (possibly non-monotone)}
We begin by discussing the situation of a generic friction force.
For simplicity, we will consider friction forces possibly multivalued only in zero, which includes all the physically relevant cases. More precisely, we consider friction density $f$ of the form
\begin{equation}
	f(t,\xi,u)=h(t,\xi,u)+q(t,\xi,u)
\end{equation}
where $h$ and $q$ are as follows: 
	\begin{enumerate}[label=\textup{(C\arabic*)},resume]
	\item \label{cond:Cjump}
	the function $h\colon \R \times \Omega \times \R\to\PP(\R)$ is of the form
		\begin{equation}
		h(t,\xi,u):=\begin{cases}
			\{-\mu^+(t,\xi)\} & \text{for $u>0$}\\
			\{\mu^-(t,\xi)\} & \text{for $u<0$}\\
			[-\mu^+(t,\xi),\mu^-(t,\xi)] & \text{for $u=0$}
		\end{cases}
	\end{equation}
	where the functions $\mu^\pm\in L^\infty(\R,L^\infty(\Omega,[0,+\infty))$ are $T$-periodic in $t$.
\item \label{cond:Csmooth} the function $q=q(t,\xi,u)\colon \R\times \Omega \times \R$ satisfies
\begin{itemize}
\item  $q(t,\xi, \cdot)$ is locally Lipschitz continuous uniformly in $(t,\xi)$;
\item for every $u\in \R$, $q(\cdot,\cdot,u)$ is measurable, bounded and $T$-periodic in $t$.
\end{itemize}
\item \label{cond:Ccoer} 
There exist $R>0$ and  two  $T$-periodic measurable functions $\ell^\pm :\R\to\R$ such that, for every $t\in \R$ 
	\begin{align*}
	\int_\Omega f(t,\xi,u) \dd\xi  &\geq -\ell^-(t) &&\text{for every  $u\leq -R$}\\
\int_\Omega f(t,\xi,u)\dd\xi	&\leq \ell^+(t) &&\text{for every  $u\geq R$}
	\end{align*} 
with
	\begin{equation*}
\int_{0}^T \left(B(s)-\ell^+(s)\right)\dd s<0< \int_{0}^T \left(B(s)+\ell^-(s)\right)\dd s \,.
	\end{equation*} 
\end{enumerate}
For such forces we have immediately the following result:

\begin{theorem}\label{th:Cgen}
	Let us consider a continuous model of crawler as above, satisfying \ref{cond:Czreg}, \ref{cond:CBreg}, \ref{cond:Cjump}, \ref{cond:Csmooth} and \ref{cond:Ccoer}. Then the same conclusions of Theorem \ref{th:Dgen} hold.
\end{theorem}
\begin{proof}
    The theorem follows from 
    Theorem~\ref{th:asymptper}.
    In particular, we notice that the maximal monotonicity required by \ref{cond:S1} follows by the convexity of $u\to \int_\Omega H(t,\xi,u+w(t,\xi))\dd\xi$, where $H(t,\xi,v)=\int_0^v h(t,\xi,u)\dd u$.
\end{proof}

\paragraph{Strictly monotone friction}
In the case of strictly monotone friction forces  we also have a result analogous to the discrete case:
\begin{theorem}\label{th:Cstrict}
	Let us consider a continuous model of crawler as above, satisfying \ref{cond:Czreg}, \ref{cond:CBreg}, \ref{cond:Cjump}, \ref{cond:Csmooth} and \ref{cond:Ccoer}. Assume moreover that $f(t,\xi,u)$ is strictly monotone decreasing in $u$ for every $t,\xi$. Then the same conclusions of Theorem \ref{th:Dstrict} hold.
\end{theorem} 
\noindent The result is a straightforward consequence of Theorem~\ref{th:Abs_strmon}.

 \smallskip

\paragraph{Dry friction}
In the continuous case, differently from the discrete one, for dry friction we have uniqueness of the periodic solution also without requiring additional regularity in time  of the shape function $z(t,\xi)$, assuming positive friction coefficients.  We prove this in a generalized case for $q$ monotone, but requiring that the inequalities in \ref{cond:Ccoer} are satisfied also for $q\equiv 0$, meaning:
\begin{enumerate}[label=\textup{(C\arabic*)},resume]
\item \label{cond:Cdry}$q(t,\xi,u)$ is monotone decreasing in $u$ for every $t,\xi$, with $q(t,\xi,0)\equiv 0$  and
	\begin{equation*}
\int_{0}^T \left(B(s)-\int_\Omega \mu^+(t,\xi)\dd \xi \right)\dd s<0< \int_{0}^T \left(B(s)+\int_\Omega \mu^-(t,\xi)\dd \xi\right)\dd s
	\end{equation*}

\end{enumerate}
and assuming that
\begin{enumerate}[label=\textup{(C\arabic*)},resume]
\item \label{cond:Cpositive}  the functions $\mu^\pm(t,\xi)$  satisfy for almost every $t$ 
\begin{equation*}
\mu^+(t,\xi)+\mu^-(t,\xi)>0    
\end{equation*}
almost everywhere in $\Omega$.
\end{enumerate}
As showed by Examples \ref{ex:dry} and \ref{ex:drystar},  analogous assumptions  in  the discrete case do not guarantee the uniqueness of the $T$-periodic orbit in that framework,   but in the continuous  setting we have the following:
\begin{theorem}\label{th:Cdry}
	Let us consider a continuous model of crawler as above, satisfying \ref{cond:Czreg}, \ref{cond:CBreg}, \ref{cond:Cjump}, \ref{cond:Csmooth}, \ref{cond:Cdry} and \ref{cond:Cpositive}.
	Then the same conclusions of Theorem \ref{th:Dmonot_reg} hold.
\end{theorem}

\begin{proof}
By  \ref{cond:Czreg}, \ref{cond:CBreg}, \ref{cond:Cjump}, \ref{cond:Csmooth} and \ref{cond:Cdry}, $G$ satisfies \ref{cond:S1}.

	We observe that, for almost every $t$, $\dot z(t,\cdot)\in W^{1,\infty}(\Omega,\R)$. Thus, for almost every $t$, we can define  the $T$-periodic functions
	\begin{equation*}
		\zeta^-(t):=\min_{\xi\in \Omega} -\dot z(t,\xi)\leq0\leq\max_{\xi\in \Omega} -\dot z(t,\xi)=:\zeta^+(t) \,.
	\end{equation*}
	The inequalities follow from the fact that $\dot z(t,\cdot)$ has zero average weighted with the measure $\rho$. Note that $\zeta^\pm(\cdot)\in L^\infty[0,T]$.
	By \ref{cond:Cdry} we have that $G(t,\cdot)$   is decreasing in $u$ and verifies
			\begin{equation}\label{ineq-}
	G(t,u)\geq -l^-_d(t):=\frac{1}{M}\left(B(t)+\int_\Omega \mu^-(t,\xi) d\xi \right )\qquad \text{for $u\leq  \zeta^-(t)$}
	\end{equation}
	and
	\begin{equation}\label{ineq+}
	G(t,u)\leq l^+_d(t):=\frac{1}{M}\left(B(t)-\int_\Omega \mu^+(t,\xi) d\xi \right )\qquad \text{for $u\geq \zeta^+(t)$\,.}
	\end{equation}
 
In particular, this implies that \ref{cond:S2} holds for $R=\max \{\norm{\zeta^+}_\infty, \norm{\zeta^-}_\infty\}$.

	Since the structural assumptions hold, by  Theorem \ref{th:Abs_weakmon} we have that  the nonempty attractor of the dynamics of equation \eqref{eq:Cgen}   is made by periodic orbits with the same derivative.
	
	We show now  that  for almost every $t$ and every  $v\in[\zeta^-(t),\zeta^+(t)]$ we have	 
	 \begin{equation}\label{eq:Cdrystima1}
	 	G(t,u_1)>G(t,u_2) \qquad \text{for every $u_1<v<u_2$\,.}
	 \end{equation}
	We consider two cases:
\begin{itemize}	
\item	If $\zeta^-(t)<\zeta^+(t)$,   we begin by  proving that $G(t, \cdot)$ is strictly decreasing on $J=[\zeta^-(t),\zeta^+(t)]$. 
	Consider  $v\in J$  and $u_1<v<u_2$.  
	By the continuity of $\dot z(t,\cdot)$   there exists an interval $I\subset \Omega$ with positive measure such that $-\dot{z}(t,\xi)\in(\zeta^-(t),\zeta^+(t))\cap (v_1,v_2) $ for any $\xi\in I$. 	Then, by  \ref{cond:Cpositive}  we have 
\begin{equation}
\int_I h(t,\xi,u_1+\dot{z}(t,\xi))\dd\xi=\int_I\mu^+(t,\xi)\dd\xi>\int_I-\mu^-(t,\xi)\dd\xi=\int_I h(t,\xi,u_2+\dot{z}(t,\xi)) \dd\xi
\end{equation}
while by the mononoticity of $h$
\begin{equation*}
\int_{\Omega\setminus I} h(t,\xi,u_1+\dot{z}(t,\xi))\dd\xi \geq \int_{\Omega\setminus I} h(t,\xi,u_2+\dot{z}(t,\xi)) \dd\xi \,.
\end{equation*}
Since also $q$ is monotone decreasing in $u$ we get \eqref{eq:Cdrystima1}. 

\item If $\zeta^-(t)=\zeta^+(t)=0$,	 and hence $\dot{z}=0$,  we get \eqref{eq:Cdrystima1} is obtained directly, since it must be verified only for $v=0$ and we know by \ref{cond:Cjump}, \ref{cond:Cdry} and \ref{cond:Cpositive} that $f(t,\xi,u_1)<f(t,\xi,u_2)$ for $u_1<0<u_2$.
\end{itemize}

	\smallskip
	
  In order to prove the uniqueness of the $T$-periodic orbit of \eqref{eq:Cgen}, we observe first that, if $\alpha(t)$ is any $T$-periodic solution of \eqref{eq:Cgen}, then it cannot be neither 
  \begin{equation}\label{less}
      \alpha(t)< \zeta^-(t)\quad \text{almost everywhere on $[0,T]$}   
  \end{equation}
  nor
	\begin{equation}\label{more}
      \alpha(t)> \zeta^+(t)\quad \text{almost everywhere on $[0,T]$.}   
  \end{equation}
	In fact, if \eqref{less} holds, we get a contradiction since
	\begin{equation*}
	    0=\int_0^T\dot{\alpha}(t) dt=\int_0^T G(t,\alpha(t)) dt\geq \int_0^T -l_d^-(t) dt>0 
	\end{equation*}
	Inequality \eqref{more} is ruled out analogously.
	
	We conclude that, given any $T$-periodic solution $\alpha(t)$ of  \eqref{eq:Cgen}, it must be
	\begin{equation}\label{middle}
	    \alpha(t)\in [\zeta^-(t),\zeta^+(t)] \qquad \text{on a subset of positive measure of $[0,T]$.}
	\end{equation}
	Assume now by contradiction that there exist two distinct $T$-periodic solutions $v^*$ and  $\bar v$  of \eqref{eq:Cgen}  with $v^*(t)<\bar v(t)$ on $[0,T]$.

	By the first part of the proof, we know that  $\dot v^*(t)=\dot {\bar v}(t)$ almost everywhere on $[0,T]$. But then,  by  \eqref{middle}   and \eqref{eq:Cdrystima1} we have 
	$$\dot{v}^*(t)\in G(t,v^*(t))>G(t, \bar v(t))\ni \dot{\bar v}(t)$$
	on a set of positive measure in $[0,T]$, a contradiction. 
	Our proof is concluded.
\end{proof}
 We remark that assumption \ref{cond:Cpositive} is necessary for uniqueness, as we show in the following example where we reconstruct the discrete model of Example~\ref{ex:dry} in a continuous setting.
 
 \begin{example}
Let us set $\Omega=[0,3]$, $\rho\equiv 1/3$, $B\equiv 0$, $q\equiv 0$ and
\begin{equation*}
\mu^\pm(t,x)=\begin{cases}
0 & \text{if $\zeta\in[1,2]$}\\
1 & \text{if $\zeta\in[0,1)\cup(2,3]$}
\end{cases}
\qquad
\dot \phi(t,x)=\begin{cases}
2 w_2(t)& \text{if $\zeta\in[1,2]$}\\
0 & \text{if $\zeta\in[0,1)\cup(2,3]$}
\end{cases}
\end{equation*}
where $w_2(t)$ is the same as in Example~\ref{ex:dry}, with $\phi(0,\cdot)\equiv 1$ and $\alpha>0$. Then we obtain the same function $G$ of  Example~\ref{ex:dry}, so that the asymptotic average velocity is not unique.
 \end{example}

\section{Discussion}
\label{sec:disc}

In this work we have studied the asymptotic gait of discrete as well as continuous crawlers with prescribed shape  moving on a line  subject to different classes of frictions and possibly to an additional external forcing, such as that due to gravity when crawling on a slope.

Our aim was to make evident how a well-posed notion of gait as asymptotic behaviour, although it can be reasonably expected, requires some special care in the mathematical formulation of the model.
For instance, triangular waves of shape change, as in Example~\ref{ex:dry}, might be a convenient choice to describe (or prescribe) the actuation in a device, or might be expected (or considered admissible) in a control problem, since they are the integral of a bang-bang strategy. We showed that this can be a problematic assumption in the case of dry friction (Examples~\ref{ex:dry} and \ref{ex:drystar}), which could however be overcome by adding a small additional viscosity to the model (Theorem~\ref{th:Dstrict}), or by assuming only smooth inputs (Theorem~\ref{th:Dmonot_reg}). Continuous models with dry friction also avoid this problem (Theorem~\ref{th:Cdry}).
We notice, at the same time, that a strong simplification of the friction forces is also not advisable, since viscosity alone is still not sufficient to proper locomotion also in a dynamic context (Example~\ref{ex:incomp}).

To better illustrate our contribution to the topic and possible future developments, we now make a brief comparison with the related results in \cite{FigKny,EldJac,ColGidVil}.
In \cite{FigKny}, the asymptotic behaviour of the discrete locomotor of Section~\ref{sec:discr} has been studied in the special case of continuous, autonomous, single-valued and strictly monotone friction forces, for a $\CC^1$ shape change, in the absence of external forces and for a more restrictive dissipative condition than \ref{cond:Dcoer}. We generalized all these assumptions and provided various examples on when they become sharp for the uniqueness of the limit behaviour. In particular, our results on dry friction show that monotonicity alone is not sufficient to prove sharp results on the uniqueness of the limit cycle, but the intrinsic structure of a crawling model (compared to a general periodically forced system) must also be exploited. 

In such framework, \cite{FigKny} also discusses the rate of convergence of the system to the limit cycle characterizing the gait. We did not address this in our work, but the topic is certainly worth of further investigation.
Indeed, we observe that, in the case of set-valued forces, and in particular of dry friction, the behaviour becomes more complex: we can observe finite-time convergence, possibly coexisting with asymptotic-only convergence for a different initial state (see e.g.~Example~\ref{ex:strib}). Such a coexistence in related problems with set-valued dissipation forces has been studied also in \cite{Cabot,ColGidVil,GudMak}. 

Two further directions of investigation are suggested by \cite{EldJac}, which studies the locomotion on the plane of a tetrahedral crawler with elastic links. All friction forces in the model are viscous. The paper proved the stability of relative periodic solutions for a small actuation with perturbative methods starting from the (stable) steady state of a passive (i.e. non-actuated) crawler. Both a more complex geometry of the crawler and an elastic body imply that the asymptotic stabilization cannot be reduced anymore to a scalar problem, but has to be studied in a higher dimensional setting, where we expect a more complex phenomenology.

To make a comparison in the case of dry friction, which is the most valuable contribution in our work, a meaningful perspective comes from \cite{ColGidVil}, where  the analogous of our discrete models of Section~\ref{sec:discr} were studied assuming an elastic body, but only for dry friction at a quasistatic regime.  At the quasistatic regime, an additional balance-breaking condition on the friction coefficients is necessary for the uniqueness of solution of the initial value problems (notice that the uniqueness argument of \cite[Theorem 2.2]{Gid18} and \cite[Example~3.2]{Gid18} can be straightforwardly adapted also to the case of prescribed shape). Yet, once uniqueness of solution is provided, uniqueness of the asymptotic average velocity of the crawler follows \cite[Theorem~11]{ColGidVil}. In this work, we show instead that, at a dynamic regime, uniqueness of solutions always holds, but uniqueness of the asymptotic average velocity for a gait is true only for smooth inputs (Theorem~\ref{th:Dmonot_reg}).

We conclude highlighting, once more, how an intuitive concept such as the characterization of a gait as an asymptotic behaviour poses theoretical issues and open problems, even in simple models. We hope that our work will clarify this picture, encourage further investigations on the topic and support the study of more complex issues, such as optimal control, which, as explained in the introduction, must be based on a well-posedness of the asymptotic behaviour.

\paragraph{Acknowledgements.} P.G. has been partially supported by the GA\v{C}R Junior Star Grant 21-09732M.
 A.M. was  supported  by FCT project  UIDB/04561/2020.
 C.R. was supported by FCT projects  UIDB/04621/2020 and UIDP/04621/2020 of CEMAT at FC-Universidade de Lisboa.



\footnotesize
\begin{thebibliography}{20}\setlength{\itemsep}{-0.5mm}
\bibitem{AdlGoe} S. Adly and D. Goeleven, A nonsmooth approach for the modelling of a mechanical rotary drilling system with friction. Evol. Equ. Control Theory 9 (2020), 915–934. 
\bibitem{Ago} D. Agostinelli, F. Alouges, A. DeSimone, Peristaltic waves as optimal gaits in metameric bio-inspired robots, Front. Robot. AI 5 (2018), 99. 
\bibitem{AkaSte} G. Akagi and U. Stefanelli, Periodic solutions for doubly nonlinear evolution equations. J. Differential Equations 251 (2011), 1790–1812. 
\bibitem{BSZZB17} C. Behn, F. Schale, I. Zeidis, K. Zimmermanna and N. Bolotnik, Dynamics and motion control of a chain of particles on a rough surface, Mech. Syst. Signal Process 89 (2017), 3–13.
\bibitem{BPZZ16} N. Bolotnik, M. Pivovarov, I. Zeidis, and K. Zimmermann, On the motion of lumped-mass and distributed-mass self-propelling
systems in a linear resistive environment. Z. Angew. Math. Mech. 96 (2016), 747–757
\bibitem{BSZZ18} N. Bolotnik, P. Schorr, I. Zeidis and K. Zimmermann, Periodic locomotion of a two-body crawling system along a straight line on a rough inclined plane. Z. Angew. Math. Mech. 98 (2018), 1930–1946.
\bibitem{Brezis} H. Brezis, Opérateurs maximaux monotones et semi-groupes de contractions dans les espaces de Hilbert, North-Holland Publishing Co., Amsterdam, 1973. 
\bibitem{Cabot} A. Cabot, Stabilization of oscillators subject to dry friction: Finite time convergence versus exponential decay results. Trans. Amer. Math. Soc. 360 (2008), 103–121.
\bibitem{ColGid} G. Colombo and P. Gidoni, On the optimal control of rate-independent soft crawlers, J. Math. Pures Appl. 146 (2021), 127–157.
\bibitem{ColGidVil}  G. Colombo, P. Gidoni and E. Vilches, Stabilization of periodic sweeping processes and asymptotic average velocity for soft locomotors with dry friction,  Discrete Contin. Dyn. Syst.  42 (2022), 737–757.
	\bibitem{Deim}  Deimling, Klaus Multivalued differential equations. De Gruyter Series in Nonlinear Analysis and Applications, 1. Walter de Gruyter \& Co., Berlin, 1992. xii+260 pp. ISBN: 3-11-013212-5
	\bibitem{DeSGidNos} A. DeSimone, P. Gidoni and G. Noselli, Liquid crystal elastomer strips as soft crawlers, J. Mech. Phys. Solids 84 (2015), 254–272. 
\bibitem{DeSTat12} A. DeSimone and A. Tatone, Crawling motility through the analysis of model locomotors: two case studies. Eur. Phys. J. E. 35 (2012).
\bibitem{EldJac}  J. Eldering and H.O. Jacobs, The role of symmetry and dissipation in biolocomotion, SIAM J. Appl. Dyn. Syst. 15 (2016), 24--59.
\bibitem{FaPaZo}  F. Fassò, S. Passarella, and M. Zoppello,
Control of locomotion systems and dynamics in relative periodic orbits,
J. Geom. Mech. 12 (2020), 395--420.
\bibitem{FedTal20} V. Fedonyuk and P. Tallapragada, Locomotion of a compliant mechanism with nonholonomic constraints.  J. Mech. Robot. 12 (2020), 051006.
\bibitem{Filip} A.F. Filippov,  Differential equations with discontinuous righthand sides. Translated from the Russian. Mathematics and its Applications (Soviet Series), 18. Kluwer Academic Publishers Group, Dordrecht, 1988. x+304 pp. ISBN: 90-277-2699-X 34-02 
\bibitem{FigKny} T. Figurina and D Knyazkov,
Periodic gaits of a locomotion system of interacting bodies.
Meccanica 57 (2022), 1463–1476. 
\bibitem{Fri} M. Frigon, M. Systems of first order differential inclusions with maximal monotone terms. Nonlinear Anal. 66 (2007),  2064–2077. 
\bibitem{Gid18} P. Gidoni, Rate-independent soft crawlers,Quart. J. Mech. Appl. Math. 71 (2018), 369--409.
\bibitem{GidDeS18} P. Gidoni and A. DeSimone, On the genesis of directional friction through bristle-like mediating elements crawler, ESAIM Control Optim. Calc. Var.  23 (2017), 1023--1046.
\bibitem{GidRiv} P. Gidoni and F. Riva, A vanishing inertia analysis for finite dimensional rate-independent systems with nonautonomous dissipation and an application to soft crawlers, Calc. Var. Partial Differential Equations,  60 (2021), art. 191.
\bibitem{GirJean} L. Giraldi and F. Jean, Periodical body deformations are optimal strategies for locomotion. SIAM J. Control Optim. 58 (2020), 1700–1714.
\bibitem{GudMak} I. Gudoshnikov, M. Kamenskii, O. Makarenkov, and N. Voskovskaia, One-period stability analysis of polygonal sweeping processes with application to an elastoplastic model, Math.Model. Nat. Phenom. 15 (2020), 25.
\bibitem{Hir} N. Hirano, Existence of periodic solutions for nonlinear evolution equations in Hilbert spaces. Proc. Amer. Math. Soc. 120 (1994), 185–192.
\bibitem{HoWi} D.G.E. Hobbelen and M. Wisse, Limit Cycle Walking. In \textit{Humanoid Robots, Human-like Machines}, edited by M. Hackel, I-Tech Education and Publishing, 2007.
\bibitem{Ijs08} A. J. Ijspeert,  Central pattern generators for locomotion control in animals and robots: a review. Neural networks, 21 (2008), 642-653.
\bibitem{KeMu} S.D. Kelly and R.M. Murray, Geometric phases and robotic locomotion. J. Robot. Syst. 12 (1995), 417-431.
\bibitem{Ken} N. Kenmochi,  Solvability of nonlinear evolution equations with time-dependent constraints and applications. Bull. Fac. Education, Chiba Univ. 30  (1981), 1-87.
\bibitem{LaugaBook} E. Lauga, The Fluid Dynamics of Cell Motility (Cambridge Texts in Applied Mathematics).  Cambridge University Press, 2020. 
\bibitem{Marvi} H. Marvi, G. Meyers, G. Russell and D. L. Hu. Scalybot: a snake-inspired robot with active control of friction, Proceedings of the ASME Dynamic Systems and Control Conference and BATH/ASME Symposium on Fluid Power and Motion Control (2012) 443–450.
\bibitem{MuHo} V.C. Müller and M. Hoffmann, What is morphological computation? On how the body contributes to cognition and control. Artificial life, 23 (2017), 1-24.
\bibitem{OdeMar} J.T. Oden and J.A.C. Martins, Models and computational methods for dynamic friction phenomena.Comput. Methods Appl. Mech. Engrg. 52 (1985),  527–634. 
\bibitem{Ort} Ortega Rafael, Periodic Differential Equations in the Plane: a topological perspective, De Gruyter Series in Nonlinear Analysis and Applications, 29. Walter de Gruyter \& Co., Berlin, 2019. xi+184 pp. ISBN: 978-3-11-005040-5
\bibitem{Ota} M. Ôtani, Nonmonotone perturbations for nonlinear parabolic equations associated with subdifferential operators, periodic problems, J. Differential Equations, 54(1984), 248–273.
\bibitem{PapRad} N.S. Papageorgiou and V.D. Rădulescu, Periodic solutions for time-dependent subdifferential evolution inclusions. Evol. Equ. Control Theory 6 (2017),  277–297. 
\bibitem{PoFeTa} B. Pollard, V. Fedonyuk, and P. Tallapragada, Swimming on limit cycles with nonholonomic constraints, Nonlinear Dyn 97 (2019), 2453--2468.
\bibitem{Rehor} I. Rehor, C. Maslen et al., Photoresponsive hydrogel microcrawlers exploit friction hysteresis to crawl by reciprocal actuation. Soft Robotics 8 (2021), 10--18.
\bibitem{VilNgu} E. Vilches and B.T. Nguyen, Evolution inclusions governed by time-dependent maximal monotone operators with a full domain. Set-Valued Var. Anal. 28 (2020), 569–581. 
\bibitem{WagLau}  G.~L.~Wagner and E.~Lauga,  Crawling scallop: friction-based locomotion with one degree of freedom, J. Theoret. Biol., 324 (2013), pp.~42--51.

\end{thebibliography}
\end{document}